\documentclass{article}
\usepackage[utf8]{inputenc}
\usepackage{xcolor}
\usepackage{amsthm}
\usepackage{amsmath}
\usepackage{amssymb}
\usepackage{graphicx}
\usepackage{float}
\usepackage{natbib}
\usepackage{comment}
\usepackage{tcolorbox}
\usepackage{bm}
\usepackage{dsfont}
\setcitestyle{authoryear,open={(},close={)}} 

\newtheorem{theorem}{Theorem}

\newtheorem{prop}{Proposition}

\begin{document}
\title{Concave-Convex PDMP-based sampling}
\author{Matthew Sutton \hspace{.2cm}\\
    Centre for Data Science, Queensland University of Technology\\
    and \\
    Paul Fearnhead \\
    Department of Mathematics and Statistics, Lancaster University\thanks{This research was supported by EPSRC grant EP/R018561}}
  \maketitle
\date{}
\maketitle

\begin{abstract}
Recently non-reversible samplers based on simulating piecewise deterministic Markov processes (PDMPs) have shown potential for efficient sampling in Bayesian inference problems. However, there remains a lack of guidance on how to best implement these algorithms. If implemented poorly, the computational costs of simulating event times can out-weigh the statistical efficiency of the non-reversible dynamics. Drawing on the adaptive rejection literature, we propose the concave-convex adaptive thinning approach for simulating a piecewise deterministic Markov process (CC-PDMP). This approach provides a general guide for constructing bounds that may be used to facilitate PDMP-based sampling. A key advantage of this method is its additive structure - adding concave-convex decompositions yields a concave-convex decomposition. This facilitates swapping priors, simple implementation and computationally efficient thinning. In particular, our approach is well suited to local PDMP simulation where known conditional independence of the target can be exploited for potentially huge computational gains. We provide an R package for implementing the CC-PDMP approach and illustrate how our method outperforms existing approaches to simulating events in the PDMP literature.
\end{abstract}

\section{Introduction}
Monte Carlo methods based on continuous-time Markov processes have shown potential for efficient sampling in Bayesian inference challenges \citep{Goldman2021,fearnhead2018piecewise}. These new sampling algorithms are based on simulating piecewise-deterministic Markov processes (PDMPs). Such processes evolve deterministically between random event times with possibly random dynamics at event times \cite[see][for an introduction to PDMPs]{PDMPDavis}. Monte Carlo methods based on simulating a PDMP with a desired target distribution of interest were first proposed in the computational physics literature where they were motivated as examples of non-reversible Markov processes \citep{peters2012rejection,michel2014generalized}. These ideas transitioned to the statistics community as an alternative to traditional Markov chain Monte Carlo (MCMC) samplers. Popular algorithms in this development include the Bouncy Particle Sampler  \citep{bouchard2018bouncy} and the Zig-Zag sampler \citep{bierkens2019zig} amongst others \citep{vanetti2017piecewise,wu2020coordinate,michel2020forward,bierkens2020boomerang}. 

There is substantial theoretical and empirical evidence to support the claim that non-reversible MCMC samplers offer more efficient sampling  as they reduce the chance of the sampler moving back to areas of the state space that have recently been visited \citep{diaconis2000analysis, bierkens2017piecewise,bouchard2018bouncy}. Samplers based on PDMPs simulate from a density $\pi(\bm{\theta})$ by introducing a velocity component $\bm{v}$ which is used to evolve the position of the state, $\bm{\theta}$, with deterministic dynamics. At stochastic event times the velocity component is updated using a Markov kernel after which the process continues with the new velocity. The main challenge to implementing a PDMP-based sampler is the simulation of the random event times. This involves simulating the next event time in a time-inhomogeneous Poisson process with rate function that depends on the current position of the sampler. 

Simulation of event times in a Poisson process is often facilitated through a method known as thinning \citep{Lewis1979}. The aim of this paper is to provide a general framework to aid practitioners in implementing PDMP-based samplers. Specifically we introduce the concave-convex thinning approach to facilitate simulating a PDMP-based sampler (CC-PDMP). This method can be applied whenever the rate function of the PDMP can be decomposed into the sum of concave and convex functions, and can be used to facilitate thinning from polynomial rate functions, allowing for broad applicability of the method. We discuss the efficiency of the method compared to alternative approaches in the literature for exact sampling of event times. 

Related to our work is the concave-convex adaptive rejection sampler of \cite{Gorur2011} which uses the same upper-bounding approach to construct bounds on the log density within a rejection sampler that draws independent samples from {univariate} density functions. In contrast, by applying this technique within PDMP sampling, the bounds are constructed on the {gradient of the log density} (i.e. the rate) and used within the non-reversible {Markov Process} framework to generate samples from {multivariate} densities. 
Moreover, the CC-PDMP framework allows for both subsampling and local updating schemes to facilitate these methods for high dimensional sampling. 
A short comparison between concave-convex adaptive rejection sampling and CC-PDMP for sampling from a univariate density is given in Appendix C. 

Some authors have proposed using automatic but approximate methods for simulating a PDMP-based sampler. Among these approaches, \cite{Cotter2020} propose numerical integration and root-finding algorithms to facilitate simulation. Others have considered using approximate local bounds which can be simulated exactly  \citep{Goldman2021, Goan2021, Pakman2017}. Both of these approaches sacrifice exact sampling of the posterior and involve a trade-off between the computational cost and the approximation of the sampling distribution.  

The paper is organised as follows. In Section 2, we introduce the technical details of simulating from a PDMP-based sampler and the details of some popular samplers. Section 3 reviews the literature on thinning for PDMP-based samplers. We introduce the concave-convex PDMP approach for adaptive thinning in Section 4. Empirical evaluation of CC-PDMP, and comparison to existing methods, is given in Section 5. We conclude with discussion of limitations and extensions that are possible for the method. Code for implementing our method and replicating our results is available at {\texttt{https://github.com/matt-sutton/ccpdmp}}.

\section{Sampling using a PDMP}
\subsection{Piecewise deterministic Markov processes}
Before we explore PDMP-based samplers, we first review PDMPs.  A PDMP is a stochastic process, and we will denote its state at time $t$ by $\bm{z}_t$. There are three key components that define the dynamics of a PDMP: 
\begin{enumerate}
\item A set of deterministic dynamics defined as $\bm{z}_{t+s} = \psi(\bm{z}_t, s)$.
\item Event times that are driven by an inhomogeneous Poisson process with rate $\lambda(\bm{z}_t)$.
\item A Markov transition kernel $q(\bm{z}|\bm{z}_{t-})$ where $\bm{z}_{t-} := \lim_{s\uparrow t} \bm{z}_s$.
\end{enumerate}
A PDMP with these specifications will evolve according to its deterministic dynamics between event times and at event times will update according to the Markov transition kernel $q(\bm{z}|\bm{z}_{t-})$. The event times must be simulated from a Poisson process with rate $\lambda(\bm{z}_t)$, which depends on the position of the process at time $t$. The result of simulating this process is a PDMP skeleton $\{(t_k, \bm{z}_{t_k})\}_{k=1}^{n}$ where for $k=1,\ldots,n$, $t_k$ denotes the $k$th event time and we have stored the state of the process immediately after this event time, $\bm{z}_{t_k}$. Given this skeleton it is possible to fill in the values of the state between the events using the deterministic dynamics.

\subsection{PDMP-based samplers}

Assume we wish to construct a PDMP process to sample from a target distribution of interest, $\pi(\bm{\theta})$. Current
PDMP-based samplers are defined on an augmented space $\bm{z} \in \mathcal{E}$, where $\mathcal{E} = \mathcal{X} \times \mathcal{V}$, that can be viewed as having a position, $\bm{\theta}_t$, and a velocity, $\bm{v}_t$. As described below, the dynamics of the PDMP can be chosen so that the PDMP's invariant distribution has the form $\nu(\bm{z}) = \pi(\bm{\theta})p(\bm{v}|\bm{\theta})$, for some conditional distribution of the velocities, $p(\bm{v}|\bm{\theta})$. As a result, the $\bm{\theta}$-marginal of the invariant distribution is the target distribution we wish to sample from. We will consider target distributions of the form
$$\pi(\bm{\theta}) \propto \exp(-U(\bm{\theta}))$$
where $\bm{\theta} \in \mathbb{R}^p$ and $U(\bm{\theta})$ is the associated potential. An important feature of PDMP samplers is that they need to know the target distribution only up to a constant of proportionality.

If we simulate the PDMP, and the process converges to stationarity, then we can use the simulated skeleton of the PDMP to estimate expectations with respect to the target distribution in two ways. Assume we are interested in a Monte Carlo estimate of $\int g(\bm{z})\nu(\bm{z})\mbox{d}\bm{z}$. The first estimator averages over the continuous time path of the PDMP
$$
\frac{1}{t_{n} - t_{1}}\sum_{k=1}^{n-1} \int_{t_{k}}^{t_{k+1}} g(\psi(\bm{z}_{t_k}, s))ds.
$$
This requires that the integral of $g(\psi(\bm{z}_{t_k}, s))$ with respect to $s$ may be computed. The second, simpler, approach is to ``discretise'' the PDMP path. This involves taking an integer $M>0$, defining $s = (t_n-t_1)/M$, calculating the value of the state at discrete-time points, $\bm{z}_{t_1+s},\ldots,\bm{z}_{t_1+Ms}$, and then using  the standard Monte Carlo estimator
$$
\frac{1}{M}\sum_{j=1}^M g(\bm{z}_{t_1 + js}).
$$
The points $\bm{z}_{t_1 + js}$ may be found by evolving the process along its deterministic dynamics for the appropriate event time interval. \\

\subsection{Bouncy particle sampler}

The Bouncy Particle Sampler \citep{bouchard2018bouncy} takes the velocity space $\mathcal{V}$ to be either $\mathbb{R}^p$ or the unit sphere $\mathcal{S}^{p-1}$ and targets an invariant distribution $\nu(\bm{z}) = \pi(\bm{\theta})p(\bm{v})$ where $p(\bm{v})$ is either the standard $p$-dimensional Normal distribution, or a uniform distribution on $\mathcal{S}^{p-1}$. This sampler evolves a PDMP with linear deterministic dynamics $\phi((\bm{\theta}_t, \bm{v}_t), s) = (\bm{\theta}_t + s\bm{v}_t, \bm{v}_t)$ between events. The canonical event rate is given by
$$
\lambda^c(\bm{z}_t) = \max\left\{0, \langle \bm{v}_t, \nabla_{\bm{\theta}}U(\bm{\theta}_t)\rangle\right\}.
$$
At an event time $t$ generated according to this rate the velocity is updated as $\bm{v}_t = R_{\bm{\theta}_{t-}}(\bm{v}_{t-})$ where
$$
R_{\bm{\theta}}(\bm{v}) = \bm{v} - 2\frac{\langle\bm{v},\nabla_{\bm{\theta}} U(\bm{\theta}) \rangle}{\|\nabla_{\bm{\theta}} U(\bm{\theta})\|^2}\nabla_{\bm{\theta}} U(\bm{\theta}).
$$
This update can be interpreted as reflecting the velocity off the hyperplane tangential to the gradient of $U(\bm{\theta})$. Additionally, at random times $t$ an additional event occurs according to a homogeneous Poisson process with rate $\lambda^{\text{ref}}>0$ in which the velocity is refreshed, drawing $\bm{v}_t\sim p(\bm{v})$. Refreshment ensures that the Bouncy Particle Sampler samples the invariant distribution and does not get ``stuck'' on contours of the potential \citep{bouchard2018bouncy}. Formally the Bouncy Particle Sampler has event rate $\lambda(\bm{z}_t) = \lambda^c(\bm{z}_t) + \lambda^{\text{ref}}$ with Markov transition kernel
$$
q(\bm{z}|\bm{z}_{t-}) = \frac{\lambda^c(\bm{z}_{t-})}{\lambda(\bm{z}_{t-})}\delta_{\bm{\theta}_{t-}}(\bm{\theta})\delta_{R_{\bm{\theta}_{t-}}(\bm{v}_{t-})}(\bm{v}) + \frac{\lambda^{\text{ref}}}{\lambda(\bm{z}_{t-})}\delta_{\bm{\theta}_{t-}}(\bm{\theta})p(\bm{v}),$$
where $\delta_a(A)$ denotes the Dirac delta function with point mass at $a$. The Bouncy Particle Sampler is an example of a \textit{global PDMP} as all components of the velocity $\bm{v}$ are changed according to the Markov transition kernel.

\subsection{Zig-Zag sampler}
The Zig-Zag sampler \citep{bierkens2019zig, Bierkens2021localZZ} takes the space of velocities, $\mathcal{V}$, to be $\{-1,1\}^p$ with $p(\bm{v}) = 2^{-p}$ uniform on this space. The Zig-Zag sampler also uses linear dynamics $\phi((\bm{\theta}_t, \bm{v}_t), s) = (\bm{\theta}_t + s\bm{v}_t, \bm{v}_t)$ between events, but unlike the BPS only a single element of the velocity vector is updated at event times. The overall Zig-Zag process can be viewed as a composition of $p$ \textit{local event rates} where each event rate has a corresponding Markov transition kernel which operates on a single component of the velocity. Specifically, the $i$-th component of the velocity is updated with local event rate 
$$
\lambda_i(\bm{z}_t) = \max\left\{0, v_i\partial_{\theta_i}U(\bm{\theta}_t)\right\}
$$
for $i = 1,...,p$. Simulating an event time in this local framework consists of simulating an event time for each local event rate and then taking the minimum time as the event time $t$ for the overall process. At an event time $t$ triggered by event rate $\lambda_i$ the velocity is updated as $\bm{v}_t = F_i(\bm{v}_{t-})$ where
$$
F_i(\bm{v}) = (v_1,\dots, v_{i-1}, -v_i, v_{i+1}, \dots, v_p)^T.
$$
Once the velocity is updated for component $i$ local event times must be re-simulated for all rates $\lambda_j(\bm{z}_t)$ which are functions of $\theta_i$. This fact can induce massive computational savings when there is conditional independence between the components of $\bm{\theta}$. To see this, consider the extreme case where the target is fully independent across dimensions. In this case each event rate, $\lambda_i(\bm{z}_t)$, will only depend on the $i$-th components of $\bm{\theta}_t$ and $\bm{v}_t$. Thus at each event, it is only the time of the next event for the component whose velocity changes that needs to be recalculated. The computational cost of simulation in this case will scale as $O(1)$ per event as opposed to $O(d)$ for the global Bouncy Particle Sampler. 

\subsection{Global and local PDMP methods}

As described above, we can divide PDMP samplers into two main classes: \textit{global methods} and \textit{local methods}. These PDMPs differ in how they operate on the velocity vector when an event is triggered. Global methods are PDMP-based methods where the transition kernel acts on the entire velocity vector. By contrast, local methods are defined so that the transition kernel only affects a subset of the velocity. In this section we will give a general algorithm for constructing a local PDMP-based sampler \cite[though see also][]{bouchard2018bouncy}. Let $\bm{v}^{[S]}$ denote the sub-vector of $\bm{v}$ with elements indexed by $S\subseteq \{1,\dots,p\}$ and $\bm{v}^{[-S]}$ denote the sub-vector of $\bm{v}$ without the elements indexed in $S$. For some set $S$ define the local kernel as one that satisfies
$$
q^{[S]}_{\bm{\theta}}(\bm{v}|\bm{v}_{t-}) = \delta_{\bm{v}^{[-S]}_{t-}}\left(\bm{v}^{[-S]}\right)q_{\bm{\theta}}^{[S]}(\bm{v}|\bm{v}_{t-}),
$$
that is, $q^{[S]}_{\bm{\theta}}(\bm{v}|\bm{v}_{t-})$ only updates the sub-vector $\bm{v}^{[S]}$. Suppose we have a partition of the components of $\theta_1, ..., \theta_p$ given by $\mathcal{S} = \{S_1, S_2, \dots, S_F\}$. Subset $S_f$ of $\bm{v}$ is updated with local kernel $q^{[S_f]}_{\bm{\theta}}(\bm{v}|\bm{v}_{t-})$ and associated event rate
$$
\lambda_f(\bm{z}_t) = \max\left\{0, \langle \bm{v}^{[S_f]}, \nabla_{\bm{\theta}^{[S_f]}}U(\bm{\theta}_t)\rangle\right\},
$$
for $f = 1,\dots, F$. Simulating a local PDMP with this structure involves simulating an event time for each of the $F$ rates and applying the local transition kernel corresponding to the smallest simulated event time. For the partition $\mathcal{S}$ define $\mathcal{N} = \{N_1,\dots, N_F\}$  with
$$
N_f = \{m \in \{1,\dots, F\} \mid \bm{\theta}^{[S_m]}   \not\!\perp\!\!\!\perp \bm{\theta}^{[S_f]} \}
$$
for $f = 1,\dots, F$ where $\bm{\theta}^{[S_f]}   \not\!\perp\!\!\!\perp \bm{\theta}^{[S_m]}$ denotes that the rate of events associated with $S_m$ depends on  some element of $\bm{\theta}^{[S_f]}$.  When an event is triggered for a local rate $f$ the velocity $\bm{v}^{[S_f]}$ is updated and new event times are simulated for rates which depend on the value of one or more $\theta_i$ for $i \in S_f$. This simulation process is summarised in Algorithm 1. 

In this framework the Zig-Zag sampler is a local PDMP with partition $\mathcal{S} = \{\{1\}, \{2\}, \dots, \{p\}\}$ and local kernels $q^{[S_f]}_{\bm{\theta}}(\bm{v}|\bm{v}_{t-}) = \delta_{F_{f}(\bm{v}_{t-})}(\bm{v})$. The local Bouncy Particle Sampler uses kernel $q^{[S]}_{\bm{\theta}}(\bm{v}|\bm{v}_{t-}) = \delta_{\bm{v}_*}(\bm{v})$ where $\bm{v}_*^{[-S]} = \bm{v}^{[-S]}_{t-}$ and $\bm{v}_{*}^{[S]} = R_{\bm{\theta}}^{[S]}(\bm{v}^{[S]}_{t-})$ with
$$R_{\bm{\theta}}^{[S]}(\bm{v}^{[S]}) = \bm{v}^{[S]} - 2\frac{\langle\bm{v}^{[S]},\nabla_{\bm{\theta}^{[S]}} U(\bm{\theta}) \rangle}{\|\nabla_{\bm{\theta}^{[S]}} U(\bm{\theta})\|^2}\nabla_{\bm{\theta}^{[S]}} U(\bm{\theta}).
$$
If there is a single factor $\mathcal{S} = \{S\}$ containing all indices $S = \{1,\dots, p\}$ we recover the global Bouncy Particle Sampler. 

\begin{tcolorbox}
\textbf{Algorithm 1: Simulating a PDMP with local structure}

\textbf{Inputs:} $t_1 = 0$, $\bm{\theta}_{t_1}, \bm{v}_{t_1}$,  factorisation $(\mathcal{S},\mathcal{N})$ with $\mathcal{S}=\{S_1,S_2, \dots, S_F \}$ and $\mathcal{N} = \{N_1,N_2,\dots, N_F\}$.\\
Sample $\tau_f$ for $f = 1,...,F$ where 
$$
\mathbb{P}(\tau_f \geq t) = \exp\left(-\int_0^t\lambda_f(\bm{\theta}_0+\bm{v}_0s, \bm{v}_0)\text{d}s \right)
$$
For $k = 1,\dots, n$ 
\begin{itemize}
\item[(a)] Let $f^* = \text{argmin}_{f} (\tau_f)$ and update $t_{k+1} = t_k + \tau_{f^*}$
\item[(b)] Update states:
$
\bm{\theta}_{t_{k+1}} = \bm{\theta}_{t_k} + \tau_{f^*}\bm{v}_{t_k}
$
\item[(c)] Update velocity:
$
\bm{v}_{t_{k+1}} = q_{\bm{\theta}_{t_{k+1}}}^{[S_{f^*}]}(\bm{v}_{t_k} |\bm{v})
$
\item[(d)] For $f \in N_{f^*}$ update times, i.e. resample $\tau_f$ where
$$
\mathbb{P}(\tau_f \geq t) = \exp\left(-\int_0^t\lambda_f(\bm{\theta}_{t_{k+1}}+\bm{v}_{t_{k+1}}s, \bm{v}_{t_{k+1}})\text{d}s \right)
$$
\end{itemize}
\textbf{Outputs:} PDMP skeleton $\{(t_k, \bm{\theta}_{t_k}, \bm{v}_{t_k})\}_{k=1}^{n}$
\end{tcolorbox}

\section{Algorithms for event-time simulation}

Simulating the first event time, $\tau$, from a Poisson process with event rate $\lambda(t)$ is equivalent to simulating $u \sim \text{Uniform}[0,1]$ and solving 
$$
\tau =  \text{argmin}_t\left\{-\log(u) =\int_0^t \lambda(\bm{z}_s)ds \right\}.
$$
If the event rate is sufficiently simple this can be done exactly. For example if the rate is constant or linear, then there are analytical solutions for $\tau$ in terms of $u$. If the rate function is convex a solution may be found using numerical methods \citep{bouchard2018bouncy}. For more complex functions, the rate can be simulated by a process known as thinning. This process makes use of Theorem 1.
\begin{theorem}[\cite{Lewis1979}]
Let $\lambda(t)$ and $\bar{\lambda}(t)$ be continuous functions where $\lambda(t) \leq \bar{\lambda}(t)$ for all $0\leq t_{\max}$. Let $\tau_1,\tau_2,..., \tau_n$ be event times of the Poisson process with rate $\bar{\lambda}(t)$ over the interval $[0,t_{\max})$. For $i = 1,...,n$ retain the event times $\tau_i$ with probability $\lambda(\tau_i)/\bar{\lambda}(\tau_i)$ to obtain a set of event times from a non-homogeneous Poisson process with rate $\lambda$ over the interval $(0, t_{\max}]$.
\end{theorem}
The efficiency of simulating via thinning depends on how tightly the rate $\bar{\lambda}$ upper-bounds $\lambda$ and how costly it is to simulate from $\bar{\lambda}$. Another key tool for enabling the simulation of event times is known as superposition and the general process is outlined in Theorem 2. 
\begin{theorem}[\cite{Kingman1992}]
Suppose that $\Lambda_1,...,\Lambda_n$ are a set of independent Poisson processes with rates $\lambda_1(t),...,\lambda_n(t)$ with first arrival times $\tau^{[1]},...,\tau^{[n]}$. The Poisson process with rate $\lambda(t) = \sum_{i=1}^n\lambda_i(t)$ may be simulated by returning the first arrival time as $\tau = \min_{i=1,...,n} \tau^{[i]}$. 
\end{theorem}
Both superposition and thinning provide useful tools in the simulation of event times from a non-homogeneous Poisson process. However, constructing the upper-bounds required for thinning is largely an ad hoc and time consuming process for the practitioner. A common desire for a practitioner is to reuse simulation knowledge. For example if an event rate can be written as the sum of several sub-event rates, that can be exactly simulated, ideally this knowledge should be used to simulate from the sum. We refer to this class of event rates as \textit{additive} rates. Specifically, we will refer to a Poisson process with a rate 
$$\lambda(\bm{z}_t) = \max\{0, f_1(t)+f_2(t)\}$$ 
as process with an additive rate. The superposition method from Theorem 2 facilitates additive event-time simulation by thinning using an event rate 
$$\bar{\lambda}(t) =\max\{0, f_1(t)\} + \max\{0, f_2(t)\},$$
which satisfies $\lambda(\bm{z}_t) \leq \bar{\lambda}(t)$. We can simulate the first event from a process with this rate as $\tau = \min(\tau_1, \tau_2)$, where $\tau_1$ is simulated with rate $\bar{\lambda}_1= \max\{0, f_1(t)\}$ and $\tau_2$ with rate $\bar{\lambda}_2 = \max\{0, f_2(t)\}$. This simulated time will be accepted with probability
$$
\frac{\max\{0, f_1(\tau)+f_2(\tau)\}}{ \max\{0, f_1(\tau)) + \max(0, f_2(\tau)\}},
$$
which will be one when both $f_1(\tau)>0$ and $f_2(\tau) >0$. This procedure is useful since it allows the re-use of thinning procedures for $f_1$ and $f_2$. Consequently a practitioner can build up to simulating from a complex rate by combining smaller simpler rates. Thinning via superposition is often recommended in the literature; notable examples include \cite{bouchard2018bouncy} who illustrate this approach when simulating from an exponential family and \cite{wu2020coordinate} who discuss the approach generally. The approach was also recommended broadly by \cite{sen2019efficient} where it was suggested for practical implementation of the Zig-Zag where $f_1$ and $f_2$ correspond to the terms from the likelihood and prior respectively. Once simulation for a choice of prior or likelihood are individually established this knowledge can be reused in future modelling.

\section{Concave-convex adaptive thinning}

Our general proposal for simulating events in a PDMP is based on concave-convex adaptive thinning. As the dynamics between events are deterministic, conditional on the current state of the PDMP, we can re-write the rate of the next event as a function of time. With slight abuse of notation we will denote this rate as $\lambda(t)$ with the argument of the function making it clear whether we are viewing it as a function of the state, or as here, as a function of time. Suppose $\lambda(t)=\max\{0,f(t)\}$ where $f(t)$ can be decomposed as:
\begin{align}
\label{eq:cc}
f(t) = f_u(t) + f_n(t)
\end{align}
on a finite interval $t \in [0, \tau_{\max})$ where $f_u(t)$ is a convex function and $f_n(t)$ is a concave function in time. (We discuss below general conditions where such a decomposition is possible.) The problem of upper-bounding $f(t)$ is recast to finding upper-bounding piecewise linear functions $\ell_u(t)$ and $\ell_n(t)$ such that $f_u(t) \leq \ell_u(t)$ and $f_n(t) \leq \ell_n(t)$. We may then apply the bound $f(t) \leq \ell(t)$ where $\ell(t) = \ell_n(t) + \ell_u(t)$. Since $\ell(t)$ is the sum of piecewise linear functions it will also be a piecewise linear function and direct simulation from a Poisson process with rate $\ell(t)$ is possible (see Appendix A).\\
Concave-convex adaptive thinning proceeds by constructing a piecewise linear function $\ell(t)$ over a set of $m$ abscissae $t_0 = 0 < t_1 < \dots < t_m = \tau_{\max}$ at which the evaluation of $f_u(t_i), f_n(t_i)$ and derivative $f_n'(t_i)$ are known. An event time $\tau$ is simulated from the Poisson process with rate $\bar{\lambda}(t) = \max\{0,\ell(t)\}$ and is accepted with probability 
$$
\frac{\max\{0, f_u(\tau) + f_n(\tau)\}}{\bar{\lambda}(\tau)}.
$$
If the event is rejected, we can reuse the information from the evaluation of $f_u(\tau)$ and $f_n(\tau)$ with an additional evaluation of $f'_n(\tau)$ to refine the simulation on the abscissae $t_0 = \tau < t_k < \dots < t_m = \tau_{\max}$ where $t_{k-1} = \max\{t_i \mid t_i < \tau\}$. If an event does not occur on the range $t\in [0,\tau_{\max})$ the PDMP process is evolved by $\tau_{\max}$ and the thinning process is repeated. We give the general construction of the piecewise linear function $\ell(t)$ for $t \in [t_1, t_2)$ and note this construction may be applied iteratively over the abscissae. This construction is also depicted visually in Figure \ref{fig:cc}.
\begin{figure}[t]
    \centering
    \includegraphics[width = .7\textwidth]{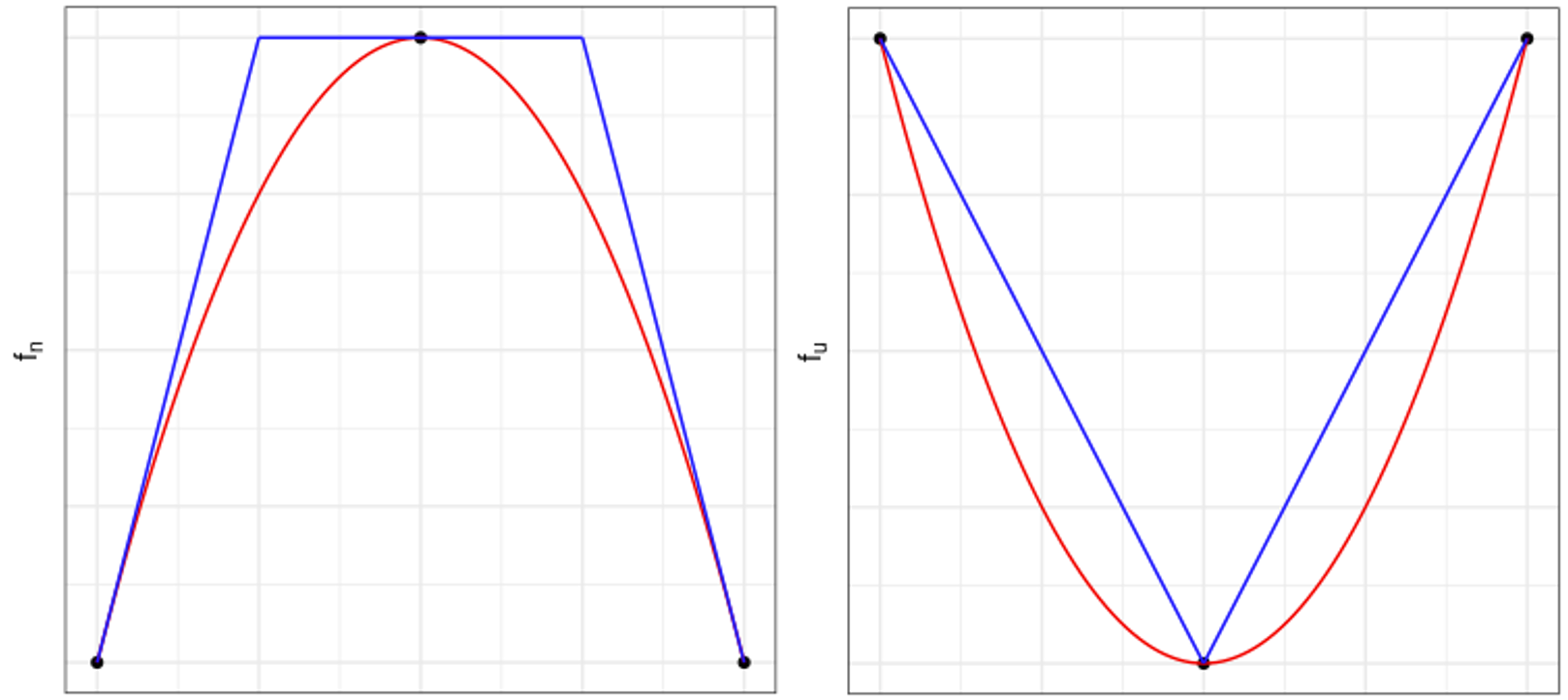}
    \caption{Upper bounds for a rate function based on concave (left) and convex (right) information on three abscissae.  Concave bounds are formed using linear segments from the gradient of $f_n(t)$. Convex bounds are constructed using linear segments connecting evaluations of $f_u(t)$.}
    \label{fig:cc}
\end{figure}
A convex function $f_u(t)$, is by definition a function that can be upper-bounded by the line segment connecting any two function evaluations. So, for any $t \in [t_1,t_2)$ we have the upper-bound $f_u(t)\leq \ell_u(t)$ where
$$
\ell_u(t) = f_u(t_1) + \frac{f_u(t_2) - f_u(t_1)}{t_2 - t_1}t.
$$
For a concave function $f_n(t)$ with derivative $f_n'(t)$ the line segment corresponding to the tangent of $f_n(t)$ will upper-bound the function. Thus $f_n(t)$ will be upper-bounded by
$$
\ell_n(t) = \min\left\{f_n(t_1) + f_n'(t_1)t, f_n(t_2) + f_n'(t_2)(t_2-t) \right\}.
$$
for $t\in[t_1,t_2)$. This minimum will switch at the point of intersection between the lines $f_n(t_1) + f_n'(t_1)t$ and $f_n(t_2) + f_n'(t_2)(t_2-t)$. It is simple to find this intersection point \citep{Gorur2011},
$$
t^* = \frac{f_n(t_2) - f_n'(t_2)t_2 - f_n(t_1) + f_n'(t_1)t_1}{f_n'(t_1) - f_n'(t_2)}
$$
which will be a point on the interval $[t_1,t_2]$ provided the derivatives $f_n'(t_1)$ and $f_n'(t_2)$ are not equal. If $f_n'(t_1) = f_n'(t_2)$, the linear segment will not change over the interval and we take $t^* = t_2$. So we take 
$$
\ell_n(t) = \begin{cases} f_n(t_1) + f_n'(t_1)t & t\in [t_1,t^*)\\
f_n(t_2) + f_n'(t_2)(t_2-t) & t \in [t^*, t_2) \end{cases},
$$
and combining these bounds we have $\ell(t) = \ell_u(t) + \ell_n(t)$
which is a piecewise linear function upper-bounding $f(t) \leq \ell(t)$ for $t\in[t_1,t_2)$.

\subsection{Concave-convex decompositions}

The proposition below gives simple conditions for the class of non-homogeneous Poisson processes that admit thinning using the concave-convex approach. 

\begin{prop}
\label{prop:W}
If a Poisson process with rate function $\lambda(t)$ can be written as ${\lambda(t) = \max\{0, f(t)\}}$ where $f(t)$ is continuous except possibly at a finite number of known points, $0 < t^{\star}_1 < t^{\star}_2 <\cdots t_r^{\star}<\infty$ , then the process admits thinning using the concave-convex adaptive thinning approach. 
\end{prop}
\begin{proof}
Consider $f(t)$ on the interval $[0,\tau_{\max})$ where if $t_1^{\star}$ exists, $\tau_{\max} = t_1^{\star}$ otherwise $\tau_{\max}$ is some arbitrary value $\tau_{\max}>0$. On this interval $f(t)$ is continuous and by the Stone-Weierstrass theorem \citep{Stone1948}, for any $\epsilon>0$ there exists a polynomial $g(t)$ on the interval $[0, \tau_{\max}]$ with $\|f - g\| < \epsilon$. The function $g$ is a polynomial so it admits the concave-convex decomposition
$$
g(t) = \sum_{\{i : a_i> 0\}}a_it^i + \sum_{\{i : a_i < 0\}}a_it^i
$$
with convex function $g_u(t) = \sum_{\{i : a_i> 0\}}a_it^i$ and concave function $g_n(t) = \sum_{\{i : a_i<0\}}a_it^i$. The Process with rate $\lambda(t)$ admits thinning on the interval $[0,\tau_{\max})$ using the Poisson process with rate $\hat{\lambda}(t) = \max\{0, g(t)+\epsilon\}$ which has a convex-concave decomposition. If an event does not occur on the interval $[0,\tau_{\max})$ then the process is considered on the next continuous interval.
\end{proof}

The proof of Proposition \ref{prop:W} relies on finding a polynomial that can upper-bound the process over a finite interval. If a bound can be given for higher order derivatives of $f(t)$ then there are two main approaches for finding an upper-bounding polynomial. Using Taylor's expansion of $f(t)$ about zero
$$
f(t) = f(0) + f'(0)t + \frac{f''(0)}{2!}t^2 + \cdots + \frac{f^{(k)}(t)}{k!}t^k + \int_0^tf^{(k+1)}(s)\frac{t^k}{k!}ds.
$$
If there is a known constant $M$ such that $|f^{(k+1)}(t)| \leq M$ we can employ thinning based on the polynomial bound $g(t) = f(0) + f'(0)t + \frac{f''(0)}{2!}t^2 + \cdots + \frac{f^{(k)}(t)}{k!}t^k + \frac{M}{(k+1)!}t^{(k+1)}$. Alternatively polynomial upper-bounds can be constructed based on interpolation where the error is controlled. For example using Lagrange polynomial interpolation on $[0,\tau_{\max}]$ with $k$ evaluations of the function has error is bounded by $\tau_{\max}^{k+1}\frac{M}{(k+1)!}$. Adding this bound as a constant yields an upper-bounding polynomial without requiring evaluation of the derivatives of $f(t)$ required for the Taylor series expansion. 

\subsection{Concave-convex adaptive thinning and PDMP-based samplers}

We will refer to the process of using concave-convex adaptive thinning for simulating a PDMP-based sampler as CC-PDMP sampling. In this section we illustrate how CC-PDMP sampling is applied when the target distribution corresponds to a Bayesian posterior:
$$
\pi(\bm{\theta}) \propto \exp(-U^{(\ell)}(\bm{\theta}) -U^{(p)}(\bm{\theta}))
$$
where $U^{(\ell)}(\bm{\theta})$ 
is the potential of the likelihood and $U^{(p)}(\bm{\theta}) = -\log p_0(\bm{\theta})$ is the potential for the prior distribution $p_0$. Consider the $k$th rate for a fully local PDMP such as the Zig-Zag sampler
$$
\lambda_{k}(t) =\max\left\{0, v_k\partial_{\theta_k}U^{(\ell)}(\bm{\theta} + t\bm{v}) + v_k\partial_{\theta_k}U^{(p)}(\bm{\theta} + t\bm{v})\right\},
$$
here we could apply concave-convex thinning with $f_k(t) = v_k\partial_{\theta_k}U^{(\ell)}(\bm{\theta} + t\bm{v}) + v_k\partial_{\theta_k}U^{(p)}(\bm{\theta} + t\bm{v})$. For Bayesian inference problems the function $f_k$ is the sum $f_k(t) = f_k^1(t) + f_k^2(t)$ of likelihood and prior part where $f_k^1(t) = v_k\partial_{\theta_k}U^{(\ell)}(\bm{\theta} + t\bm{v})$ and $f_k^2(t) = v_k\partial_{\theta_k}U^{(p)}(\bm{\theta} + t\bm{v})$. If $f_k^1$ and $f_k^2$ have concave-convex decompositions then $f_k$ has a concave-convex decomposition. This allows for reuse of bounding knowledge, priors and likelihoods can be swapped provided their decompositions are known. Moreover, to implement this fully local method the rates $\lambda_k(t) = \max\{0,f_k(t)\}$ must have a concave-convex decomposition for $k=1,...,p.$ Thus we can reuse this information to find a decomposition for a global PDMP-based sampler with rate 
$$
\lambda(t) = \max\{0, \langle \bm{v}, \nabla_{\bm{\theta}}U(\bm{\theta}+t\bm{v})\rangle\} = \max\left\{0, \sum_{k=1}^pf_k(t)\right\}.
$$
Terms can be trivially regrouped to facilitate simulation from any choice of local PDMP structure as described in Section 2.5. 

\section{Tuning parameters}

\subsection{Choice of abscissae}

A key user choice in CC-PDMP sampling is that of the position and number of abscissae on the interval $[0,\tau_{\max})$.
Suppose that CC-PDMP simulation is implemented with abscissae $t_0=0<t_1<\cdots<t_n=\tau_{\max}$ if an event does not occur on the interval $[0,\tau_{\max})$ the PDMP will be evolved by $\tau_{\max}$ and the thinning process will repeat where the function evaluations at $\tau_{\max}$ can be reused at $t_0=0$. Thus one simulation with $n$ abscissae would be computationally equivalent to $n$ evolutions of the CC-PDMP approach using abscissae with only two points $t_0=0$ and $t_1 = \tau_{\max}$. This encourages choosing a minimal number of abscissae and more carefully choosing the length of the interval. 

There are two main issues that can occur when tuning the parameter $\tau_{\max}$. If $\tau_{\max}$ is too small, then the proposal frequently lie outside the interval $[0,\tau_{\max})$ and simulating an event will requiring many iterations of bounding the event rate. Whilst if $\tau_{\max}$ is too large then the bound on the rate may be be poor, leading to many events that are rejected. In this article, we refer to an iteration of the CC-PDMP that does not generate an event time (i.e. a rejected proposal event or not generating on the interval) as a \textit{shadow event}. The total number of iterations will be the sum of the number of events and shadow events. \textit{Efficiency} is measured as the proportion of iterations that are events not shadow events, calculated as the ratio of the number of events to the number of iterations.

Adapting the value of $\tau_{\max}$ will not change the sampling dynamics. Consequently, unlike in adaptive MCMC, this parameter may be adapted throughout the course of simulating the PDMP. Ideally this parameter should be a little larger than the average event time so that events are proposed on the interval. A simple automatic approach for choosing this parameter is to set $\tau_{\max}$ equal to the $q$-th percentile of previous simulated event times, and update $\tau_{\max}$ every 100 iterations of the sampler. We found that setting $q=80$th percentile worked well for automatically selecting $\tau_{\max}$.

We investigate the dependency on this tuning parameter by implementing the Zig-Zag sampler on the Banana distribution. The 2-dimensional Banana distribution has potential:
$$
U(\bm{\theta}) = (\theta_1 - 1)^2 + \kappa(\theta_2 - \theta_1^2)^2
$$
where $\kappa$ controls how much mass concentrates around the region $\theta_2 \approx \theta_1^2$.

Since the potential is polynomial in both parameters it is trivial to see that the Zig-Zag rate functions will also be polynomial. Specifically the rates will be polynomial functions of time $t$ of degree 3 and 2 for parameters $\theta_1$ and $\theta_2$. Further implementation details may be found in the supplementary material and associated GitHub. 

Figure \ref{fig:tuneTmax} shows the efficiency of the Zig-Zag sampler on the Banana distribution with $\kappa = 1$ when changing the interval length $\tau_{\max}$ from 0.1 to 4 with two abscissae. The Zig-Zag algorithm was simulated for 10000 events at each value of $\tau_{\max}$. The effect of selecting $\tau_{\max}$ too small or too large is clearly seen and the red line shows the performance using the adaptive approach. While it is possible to get low efficiency if the parameter $\tau_{\max}$ is chosen poorly, we have found either using the adaptive approach or a short pilot run of the sampler to tune $\tau_{\max}$ gives reliable performance. 

\begin{figure}[H]
\centering
\includegraphics[width = 0.9\textwidth]{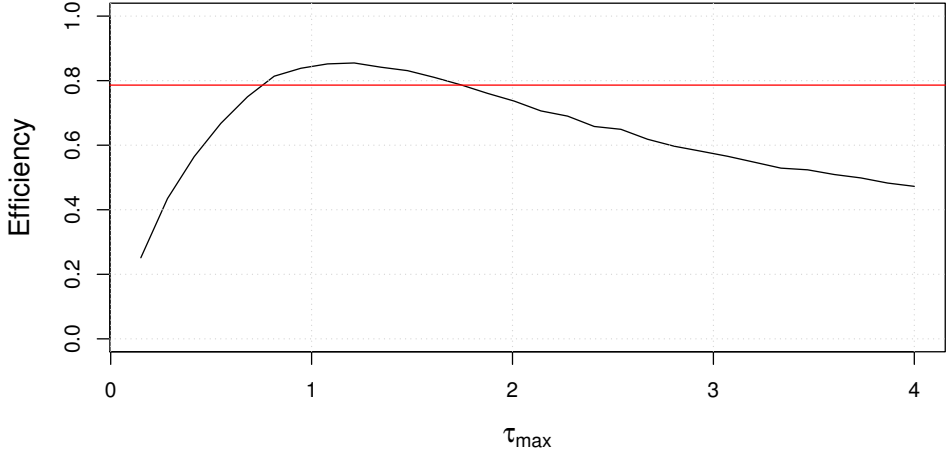}
\caption{Efficiency, defined as the proportion of events that are not shadow events, of Zig-Zag using the CC-PDMP thinning with varying choices for $\tau_{\max}$ on the Banana target. The red line shows the performance of an automatic choice for this parameter.}
\label{fig:tuneTmax}
\end{figure}

\subsection{Choice of decomposition}
The choice of concave-convex decomposition is not unique \citep{Gorur2011}. Arbitrary convex functions can be added to $f_u(t)$ and subtracted from $f_n(t)$ to give a new valid decomposition. A concave-convex decomposition is minimal if there is no non-affine convex function that can be added to $f_n(t)$ and subtracted from $f_u(t)$ while preserving the convexity of the decomposition \citep{Hartman1959}. Consider the function 
$$f(t) = -t^3+3t^2-3t+3,$$
where $f''(t) = -6t+6$. It is simple to see that the function is convex for $t<1$ and concave for $t > 1$. A minimal decomposition is given by the piecewise functions 
$$f_u^1(t) = \begin{cases}-t^3+3t^2-3t+3 & t \leq 1\\ 0 & t > 1 \end{cases} \text{ and } f_n^1(t) =\begin{cases}0 & t \leq 1\\ -t^3+3t^2-3t+3 & t > 1 \end{cases}.$$
\cite{Gorur2011} give a general construction for a minimal concave-convex decomposition. However, this construction relies on finding all points of inflection, which are points where the function changes convexity. Alternatively, Proposition 1 gives a simple decomposition $f_u^2(t) = 3t^2+3$ and $f_n^2(t) = -t^3-3t$. While the decomposition from Proposition 1 is not minimal, it does not require finding all points of inflection. These two decompositions are shown in Figure \ref{fig:tunepoly} using abscissae at 0 and 1.  The optimal decomposition gives a tighter bound, here reducing the bounding rate by approximately 0.33 on average across the interval. However, using this bound comes with additional computation cost of finding inflection points. This has the potential to reduce the overall efficiency of the method. Our experience is that the efficiency gains for using the optimal polynomial are generally not sufficient for it to be beneficial. 

\begin{figure}[H]
\centering
\includegraphics[width = 0.95\textwidth]{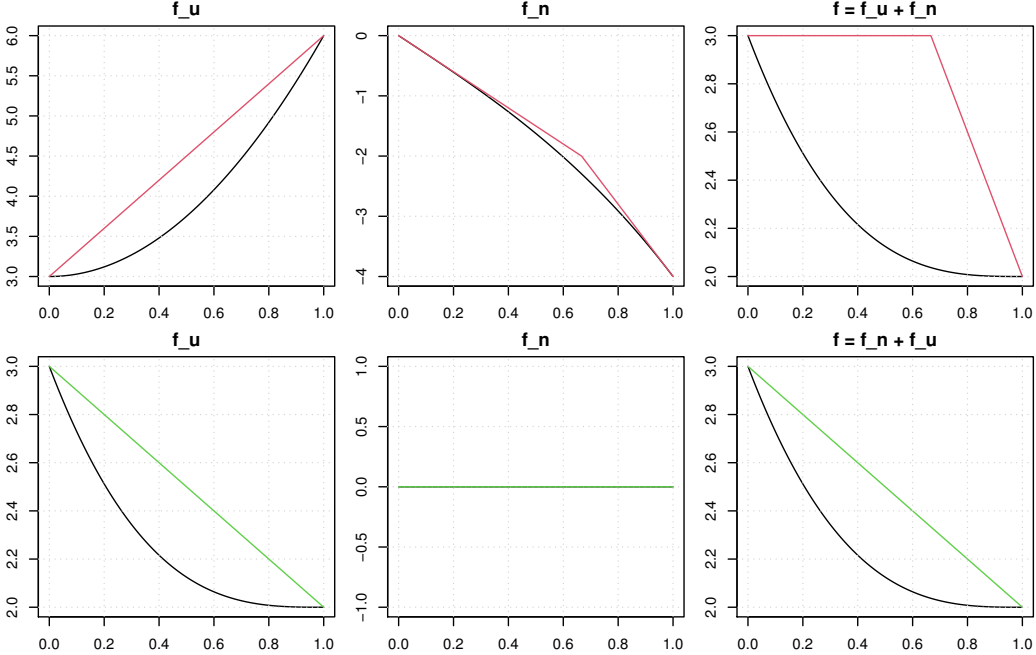}
\caption{The upper-bounds resulting from two different concave-convex decompositions of the function $f(t)=-t^3+3t^2-3t+3$, based on abscissae at 0 and 1. The top row corresponds to the decomposition from Proposition 1 and the bottom row corresponds to an optimal decomposition which recognises $f(t)$ as a convex function on $[0,1)$. The columns show the functions $f_u$, $f_n$ and $f$ as well as their piece-wise linear upper-bounds.}
\label{fig:tunepoly} 
\end{figure}

\section{Experiments}

We now present empirical evaluation of CC-PDMP and comparison with other approaches to simulate PDMPs.
Our experiments were implemented using the R package {\sc ccpdmp} available at  {\texttt{https://github.com/matt-sutton/ccpdmp}}.

The package enables you to simulate a PDMP provided one specifies the concave and convex decomposition of the rate function. If the rate function is polynomial (or bounded by a polynomial) the practitioner may parse this instead and the concave-convex decomposition will be handled internally. The package contains some basic example use cases and code to reproduce the experiments. Code snippets are given in the additional material. 

\label{exampleSection}

\subsection{Application in generalised linear models}
Generalised linear models (GLMs) provide a rich class of models that are frequently used in Bayesian inference. Let $\{(y_i, \bm{x}_i)\}_{i=1}^n$ be the observed data, where $y_i$ is an observed response and $\bm{x}_i\in \mathbb{R}^p$ is a vector of associated covariates for $i=1,\dots, n$. The expected value of $y_i$ is modelled by $g^{-1}(\bm{x}_i^T\bm{\theta})$ where $g^{-1} : \mathbb{R}\rightarrow \mathbb{R}$ is the inverse link function. The potential of the likelihood has the form
$$
U^{(\ell)}(\bm{\theta}) = \sum_{i=1}^n \phi(\bm{x}_i^T\bm{\theta}, y_i)
$$
where $\phi : \mathbb{R}^2 \rightarrow \mathbb{R}$ is the GLM mapping function $\phi(a, y) = \log p(y | g^{-1}(a))$ returning the log likelihood for observation $y$ given $a$. The partial derivative with respect to $\theta_k$ has the form
$$
\partial_{\theta_k} U^{(\ell)}(\bm{\theta}) = \sum_{i=1}^n \phi'(\bm{x}_i^T\bm{\theta}, y_i)x_{ik}
$$
where $\phi'(a, y) = \frac{\partial}{\partial a}\phi(a, y)$ and higher order derivatives are defined similarly. Over the time interval $t \in [0, \tau]$ let $f_k(t) = v_k \partial_{\theta_k} U^{(\ell)}(\bm{\theta} + t\bm{v})$ which is
$$
f_k(t) = v_k\sum_{i=1}^n \phi'(a_i(t), y_i)x_{ik}
$$
where $a_i(t) = \bm{x}_i^T(\bm{\theta} + t\bm{v})$. We can use this to define local rates as $\lambda_k(t) = \max\{0, f_k(t)\}$ for $k = 1,...,p$. For GLMs, repeated application of the chain rule yields:
$$
f_k^{(j)}(t) = v_k\sum_{i=1}^n \phi^{(j+1)}(a_i(t), y_i)\left(\frac{\partial a_i}{\partial t}\right)^{j}x_{ik}
$$
where $\frac{\partial a_i}{\partial t} = \bm{x}_i^T\bm{v}$ and $f_k^{(j)}$ denotes the $j$th derivative of $f_k(t)$. When there are bounds on $\phi^{(j)}$ we can use the upper-bounding Taylor polynomial. In Appendix B we provide bounds for several modelling choices. 

Here we consider a logistic regression GLM which has the mapping function $\phi(a, y) = \log(1 + \exp(a)) - ya$ where $y \in\{0,1\}$. We look at the efficiency of CC-PMDP thinning for a three dimensional logistic regression problem with $n=200$ observations. The covariates were generated from a multivariate normal, $\bm{x}_i \sim \mathcal{N}(\bm{0},V^{-1})$ for $i=1,..,200$, with mean zero, precision matrix 
$$
V = \left(\begin{matrix}
1 & \rho & 0 & 0 & 0\\
\rho & 1 & 0 & 0 & 0\\
0 & 0 & 1 & 0 & 0\\
0 & 0 & 0 & 1 & 0\\
0 & 0 & 0 & 0 & 1\\
\end{matrix}\right)
$$
and data generated using $\theta = (-1.25,0.5,-0.4,-0.4,-0.4)^T$. We take a Gaussian prior $\theta_j \sim \mathcal{N}(0, 1)$ for $j = 1,\dots,5.$ As the correlation is increased the thinning becomes more challenging. We investigate the performance of CC-PDMP for increasing polynomial bounds with increasing $\rho$.

\begin{table}[H]
\centering
\begin{tabular}{cccccccc}
  \hline
   & \multicolumn{7}{c}{Correlation $(\rho)$}\\
Polynomial Order & 0.00 & 0.25 & 0.50 & 0.65 & 0.75 & 0.85 & 0.95 \\ 
\hline
1 & 0.53 & 0.50 & 0.45 & 0.39 & 0.34 & 0.27 & 0.15 \\ 
  2 & 0.80 & 0.80 & 0.79 & 0.78 & 0.76 & 0.71 & 0.46 \\ 
  3 & 0.82 & 0.82 & 0.82 & 0.82 & 0.81 & 0.79 & 0.62 \\ 
   \hline
\end{tabular}
\caption{Efficiency of thinning in Zig-Zag sampling of logistic regression for increasing order of Taylor series polynomial thinning bound. Efficiency is measured as the average proportion of iterations that are events not shadow events over 20 repetitions of the sampler.}
\label{tabRes1}
\end{table}

To date only linear bounds for logistic regression have been used for thinning \citep{bierkens2018piecewise,bouchard2018bouncy}. These are based on the bound $|\phi''(a, y)| \leq 1/4$. Table \ref{tabRes1} shows the efficiency for thinning using polynomials of order 1-3 (using bounds $|\phi''(a, y)| \leq 1/4$, $|\phi'''(a, y)| \leq 1/(6\sqrt{3})$ and $|\phi^{(4)}(a, y)| \leq 1/8$). The linear 1st order bound matches the bound used in \cite{bierkens2018piecewise} for logistic regression using the Zig-Zag sampler. Table \ref{tabRes1} shows that higher order polynomial thinning facilitated through the CC-PDMP allows more efficient thinning. For the linear bounds (Polynomial order 1) we used a fixed $\tau_{\max} = 1$ as linear proposals are exactly simulated using the concave-convex approach. For other polynomial orders we used the adaptive procedure described in the tuning section to select $\tau_{\max}$. 

\subsection{Comparison to thinning via superposition}

In this example we demonstrate the advantages of CC-PDMP when the event rate is additive. In particular, we compare thinning using the CC-PDMP approach with thinning using superposition as outlined in Section 3. Let $\{y_i\}_{i=1}^n$ be the observed data with the following model
\begin{align*}
y_i \mid \theta_i &\sim \text{Poisson}(\exp(\theta_i)), \\
\theta_i &\sim \mathcal{N}(0, 1)
\end{align*}
independently for $i = 1...,n$. The derivative of the potential is $\partial_{\theta_k} U(\bm{\theta}) = \theta_k - y_k + \exp(\theta_k)$ and
$$
f_k(t) = v_k\partial_{\theta_k} U(\bm{\theta} + t\bm{v}) = v_k(\theta_k + v_kt) - y_kv_k + v_k\exp(\theta_k + v_kt).
$$
Which has the concave-convex decomposition $f_u(t) = f_k(t)$ when $v_k >0$ and $f_u(t) = v_k(\theta_k + v_kt) - y_kv_k$, $f_n(t) = v_k\exp(\theta_k + v_kt)$ when $v_k < 0$. The global Bouncy Particle Sampler rate can be defined as 
$$\lambda(t) = \max\left\{0, \sum_{k=1}^n v_k\partial_{\theta_k} U(\bm{\theta} + t\bm{v})\right\} = \max\left\{0, \sum_{k=1}^n f_k(t)\right\},$$
which has concave-convex decomposition defined by the decompositions of the individual $f_k$ functions. In contrast, thinning via superposition involves simulating an event time for the linear and exponential terms of the rate individually. The proposed event time is the minimum of all simulated times. This approach to simulating event times for exponential likelihood and Poisson-Gaussian Markov random fields has been previously used in implementing PDMPs \cite{bouchard2018bouncy, vanetti2017piecewise}. 

\begin{figure}[H]
\centering
\includegraphics[width = 0.9\textwidth]{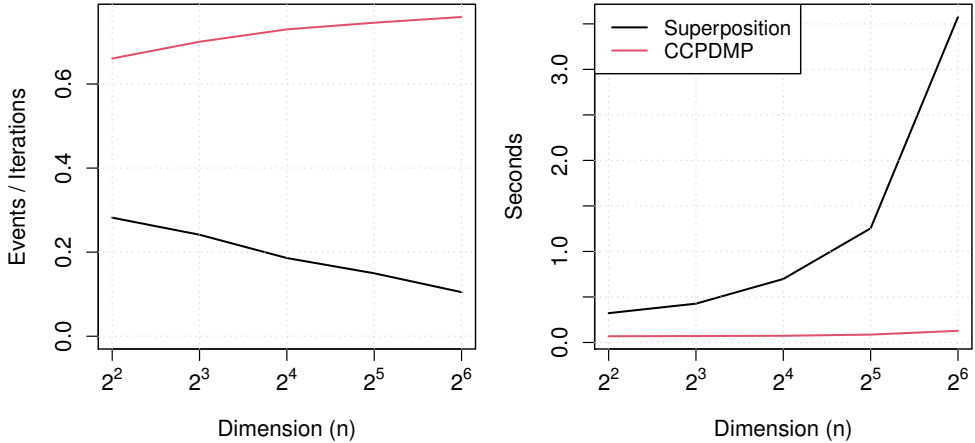}
\caption{Efficiency of thinning for the global BPS using thinning via super-position and CC-PDMP thinning with increasing dimension. Efficiency of the thinning (left) and total computation time (right) is averaged over 20 repeated runs. The Bouncy Particle Sampler samplers were simulated for 1000 event times on each run.}
\label{fig:supcompare}
\end{figure}

Figure \ref{fig:supcompare} compares the thinning and overall computational efficiency for the Bouncy Particle Sampler applied to the Poisson likelihood problem. In the CC-PDMP approach we used the adaptive update for $\tau_{\max}$ as described in the tuning section. As the dimension increases we see the proportion of iterations that are events using superposition drops quickly, consequently the overall computation time for this approach scales poorly. The poor performance of the superposition approach with increasing dimension is the result of the large number of exponential terms, $v_k\exp(\theta_k + v_kt)$, in the event rate. The thinning acceptance rate for this proposal is
$$
\frac{\max\left\{0, \sum_{k=1}^n v_k(\theta_k + v_kt - y_k + \exp(\theta_k + v_kt))\right\}}{\max\left\{0, \sum_{k=1}^n(v_k(\theta_k + v_kt) - y_kv_k) \right\} + \sum_{k=1}^n\max\left\{0,v_k\exp(\theta_k + v_kt) \right\}}
$$
When $v_k<0$ these exponential terms contribute zero to the denominator of the superposition approach. In comparison the CC-PDMP approach uses a linear upper-bound on these exponential terms which can be negative and reduce the denominator term, leading to more efficient thinning proposals. This is seen in Figure \ref{fig:supcompare} where the thinning efficiency remains roughly constant with increasing dimension. 

\subsection{Local Methods}
Local PDMP methods take advantage of the conditional independence between parameters. Consider the following extension to the previous Poisson example  
\begin{align*}
y_i \mid \theta_i &\sim \text{Poisson}(\exp(\theta_i)) \qquad \text{independently for }i=1,\dots, n \\
\theta_1 &\sim \mathcal{N}(0, 1/(1 - \rho^2))\\
\theta_i \mid \theta_{i-1} &\sim \mathcal{N}(\rho\theta_{i-1}, 1), \qquad \text{for }i=2,\dots, n
\end{align*}
where we fix $\rho = 0.5$. The prior on $\bm{\theta}$ corresponds to an AR(1) process. 
The partial derivative is
$$
\partial_{\theta_k} U(\bm{\theta}) = \begin{cases} (1 + \rho^2)\theta_k - \rho\theta_{k+1} - y_k + \exp(\theta_k) & k = 1\\
(1 + \rho^2)\theta_k - \rho\theta_{k-1} - \rho\theta_{k-1} - y_k + \exp(\theta_k) & 1<k<n\\ 
\theta_k - \rho\theta_{k-1} - y_k + \exp(\theta_k) & k = n
\end{cases}
$$
which has the same linear and exponential form as the rate in Section 5.2 so we can use an analogous concave-convex decomposition for this rate. In this section we consider local PDMP implementations. The CC-PDMP implementation facilitates simple construction of thinning bounds for local PDMP factorisations. For the partition $\mathcal{S} = \{S_1,\dots, S_F\}$ the local rate for factor $f$ is
$$
\lambda_f(\bm{z}_t) = \max\left\{0, \sum_{k \in S_f} v_k\partial_{\theta_k}U(\bm{\theta}_t)\right\}= \max\left\{0, \sum_{k \in S_f} f_k(t)\right\},
$$
for $f = 1, \dots, F$. In this section we consider a number of local decompositions of the form $\mathcal{S}^{(j)} = \{S_1,\dots, S_{F}\}$ where $S_1 = \{1,2,...,j\}$, $S_2 = \{j+1,j+2,..,2j\}$, $\dots, S_{F} = \{p-j,\dots,p-1,p\}$. If $p\neq Fj$ then the final factor will have fewer elements than the previous factors. For these decompositions the conditional independence between parameters gives the following neighbours $\mathcal{N} = \{N_1, \dots, N_F\}$ where 
$$N_f =\begin{cases} \{f, f+1\} & f = 1\\ \{f-1, f, f+1\}& 1 < f < F\\ \{f-1, f\}& f = F\end{cases}.$$
This local decomposition offers computational advantages since updating a single rate requires resimulation of three additional rates regardless of the dimension of the problem. However, this computational efficiency comes at a loss of statistical efficiency. In particular, a global PDMP will be able to move further in stochastic time than local methods before requiring a new event simulation as the global rate function will lower bound the local rate. The overall efficiency of the PDMP method should therefore be considered in terms of both its computational and statistical efficiency. 

In Figure \ref{fig:local} we investigate the scaling performance of local methods in comparison with alternative state-of-the-art MCMC methods. In particular, we compare local PDMP methods with well-tuned implementations of Metropolis adjusted Langevin algorithm (MALA) and Hamiltonian Monte Carlo (HMC). MALA is tuned to obtain acceptance rate approximately equal to 0.5 and HMC tuned to scale with acceptance probability approximately equal to 0.6. For MALA this involves scaling the variance in the proposal and for HMC this involves adjusting the number of leapfrog steps per iteration. These implementations are in line with theoretical scaling results when the parameters of the distribution are all independent. 
\begin{figure}[t]
\centering
\includegraphics[width=\textwidth]{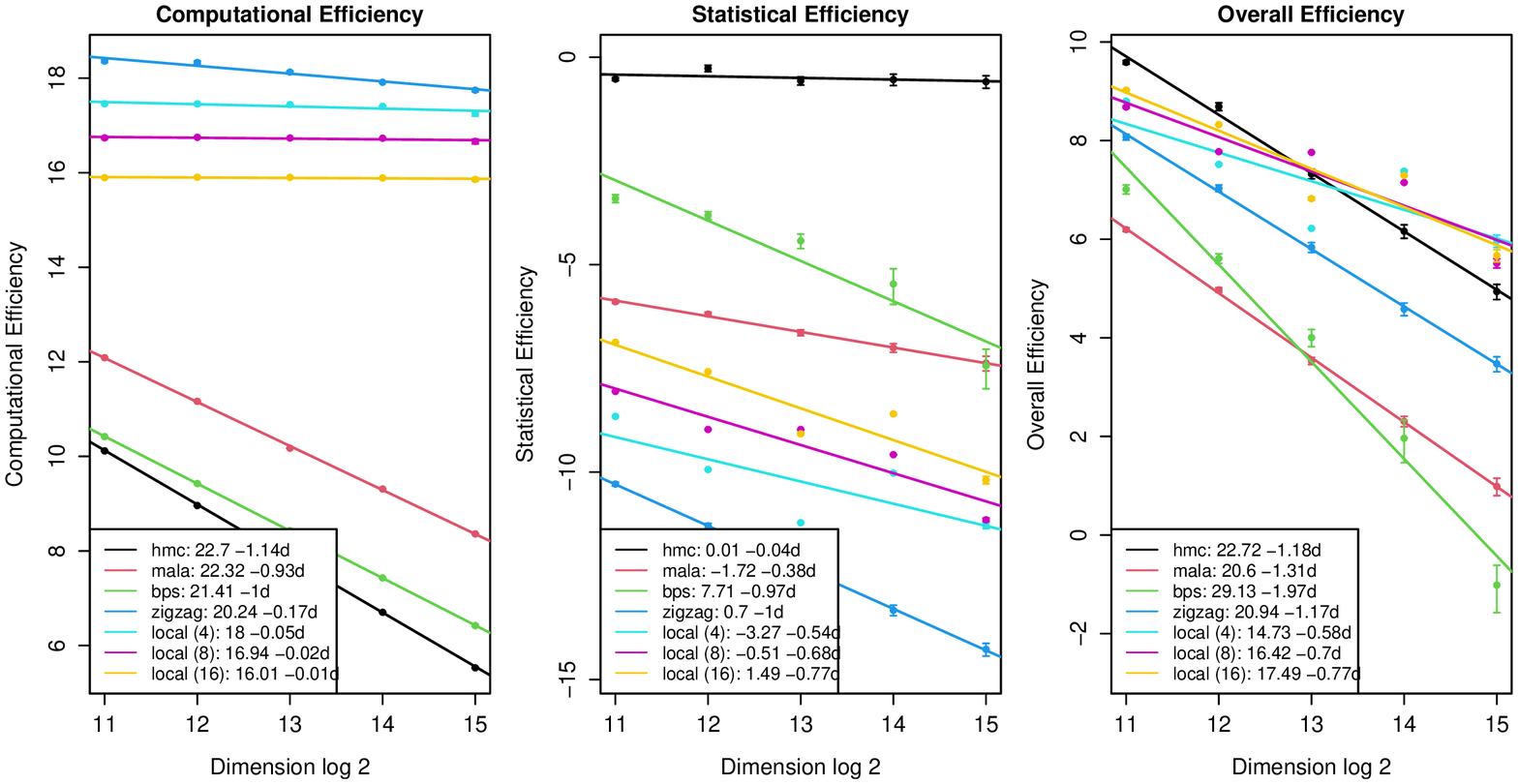}
\caption{Log-log breakdown of computational, statistical and overall empirical efficiency scaling with dimension. The ESS values are calculated with respect to the first coordinate $\theta_1$ using the coda R package on a discretised trajectory of the PDMPs. Plotted are the average rates and error bars of all methods calculated from 50 repeated runs of the methods. The legend shows the slope and intercept fitted for each method giving empirical evidence for scaling rates.}
\label{fig:local}
\end{figure}

As expected, the computational efficiency for the local methods remains (approximately) constant with increasing dimension. Based on the fitted models shown alongside the legends, it appears that computational cost for both MALA and Bouncy Particle Sampler scales approximately as $O(d^{-1})$. The computational cost for HMC scales at a rate greater than $O(d^{-1})$ due to the increase in the number of leapfrog iterations. The statistical efficiency for the local methods improves with larger local factors and is highest for the Global Bouncy Particle Sampler. The statistical efficiency for MALA drops at a rate roughly equal to $O(d^{-1/3})$. Overall efficiency can be seen as the sum of the log computational and log statistical efficiencies. The Zig-Zag and local Bouncy Particle Sampler methods attain the best overall efficiency scaling rates around $O(d)$ or better. It is clear that there is a trade-off between the statistical efficiency of larger local factors and the increased computational cost. This decomposition can be thought of as an additional choice for PDMP implementation that can easily be tuned using the CC-PDMP approach for thinning. Despite a poorer scaling with dimension, HMC remains competitive with local methods to a very high dimension. 

\section{Conclusion}

PDMP-based samplers have shown advantages over traditional MCMC samplers, but their use has been limited by the perceived difficulty in simulating the PDMP dynamics. We have introduced CC-PDMP as a general approach that can simulate the PDMP provided we can specify a concave-convex decomposition of the event rate. This method has broad applicability, enables simple implementation of local PDMP methods, and empirically outperforms alternative simulation approaches. Additional generalisations of the CC-PDMP approach are also possible by making use of other ideas for adaptive rejection sampling. For example, to bound the concave term, $f_n(t)$, of the rate function we require that the derivative $f_n'(t)$ is known on the interval $t\in[0,\tau_{\max})$. However, we could instead use a looser bound that does not require this derivative information based on the adaptive rejection bounds proposed in \cite{Gilks1992gf}. If the rate function is particularly computationally intensive, a lower bound based on the abscissae where $f_u, f_n$ and $f'_n$ and $f_u'$ are evaluated may be used to perform early rejection in the thinning. Additional generalisations may also be found in the difference of convex functions programming literature \citep{Le_Thi2018}. A useful approach to finding a decomposition was to construct a polynomial approximation that upper-bounds the rate. An approximate version of our algorithm could also be implemented where this polynomial approximation is estimated via interpolation. This would allow thinning without having to find the concave-convex decomposition. However, the overall algorithm would be biased unless error in the interpolation is added to the polynomial rate. 

Finally while our CC-PDMP approach has been illustrated only on PDMPs with linear dynamics the approach could be used more generally on PDMPs with nonlinear dynamics, such as the Boomerang sampler of \cite{bierkens2020boomerang}. The only requirement is that the rate function can be bounded by a function with a concave-convex decomposition. 


\bibliographystyle{royal}
\bibliography{ref}

\pagebreak
\appendix

\section{Supplementary Material}

\subsection{Simulating from a piecewise linear rate}
\label{subsec:linrate}
Suppose the rate function $\lambda(t) = \max\{0, \ell(t)\}$ is piecewise linear on times $t_0, t_1, \dots, t_n$ with  
\begin{align}
\label{linpiece}
\ell(t) = \begin{cases}a_0 + b_0t & 0<t<t_1\\a_1 + b_1t & t_1<t<t_2\\ \dots& \dots\\a_{n-1} + b_{n-1}t & t_{n-1}<t<t_n  \end{cases}
\end{align}
The event time can be simulated using the following process:
\begin{tcolorbox}
Simulate $u\sim $Uniform[0,1]\\
Set $E = -\log(u)$\\
For $k =0,...,n-1$:\\
1. If $\int_{t_k}^{t_{k+1}} \max\{0, a_k + b_ks\} ds > E$, then return $\tau$ solving:
 $$
 \tau =  \text{argmin}_t\left\{E =\int_{t_k}^t \max(0, a_k + b_ks) ds \right\}.
 $$
2. If $\int_{t_k}^{t_{k+1}} \max\{0, a_k + b_ks\} ds \leq E$, then update $E$ as
$$
E = E - \int_{t_k}^{t_{k+1}} \max\{0, a_k + b_ks\} ds
$$
\textbf{Return }
\end{tcolorbox}

\section{Alternative GLM likelihood and prior specifications}
The likelihood potential for a GLM posterior has the form
$$
U^{(\ell)}(\bm{\theta}) = \sum_{i=1}^n \phi(\bm{x}_i^T\bm{\theta}, y_i)
$$
where $\phi : \mathbb{R}^2 \rightarrow \mathbb{R}$ is the GLM mapping function $\phi(a, y) = \log(y | g^{-1}(a))$. As described in Section 5.1 polynomial rates can be constructed if the following functions can be evaluated and bounded
$$
f_k^{(j)}(t) = v_k\sum_{i=1}^n \phi^{(j+1)}(a_i(t), y_i)\left(\frac{\partial a_i}{\partial t}\right)^{j}x_{ik}.
$$
Functions and derivatives can be found using mathematical software such as Mathematica. 
\subsection{Robust regression}
For robust regression of a response $y\in \mathbb{R}$ we could model the error with wide tails using the \textbf{Cauchy likelihood:}
$$
\phi(a, y) = -\log\left(1 + \frac{(a-y)^2}{b}\right)
$$
with scale $b > 0$. The derivatives are $\phi'(a, y) = -2\frac{a-y}{b + (a-y)^2}$, $\phi''(a, y) = 2\frac{(a-y)^2 - b}{(b + (a-y)^2)^2}$ and $\phi'''(a, y) = -4\frac{(a-y)((a-y)^3 - 3b)}{(b + (a-y)^2)^3}$ with upper-bounds $|\phi'(a, y)|\leq \frac{1}{\sqrt{b}}$, $|\phi''(a, y)|\leq \frac{2}{b}$. For $b \geq 1$ $|\phi'''(a, y)|\leq 3$ (though better bounds are available). For $b \leq 1$, the gradient can be bounded but closed form solutions for this bound are not available. \\
\textbf{Mixture of normals:}
$$
\phi(a, y) = -\log\left( 0.5\mathcal{N}(y-a; 0, 1) +0.5\mathcal{N}(y-a; 0, 10^2) \right),
$$
upper-bounds on the derivatives for this function are $|\phi'(a, y)|\leq 2$, $|\phi''(a, y)|\leq 0.91$ and $|\phi''(a, y)|\leq 1.8$. 
\subsection{Poisson regression}
For $y \sim $Poisson($\exp(a)$) the mapping function is,
$$
\phi(a, y) = -\log\left(\exp(ya - \exp(a)) \right).
$$
Let $a^* = \max\{a(0), a(\tau_{\max})\}$ where $a(t) = \bm{x}^T(\bm{\theta} + t\bm{v})$. The derivatives for the mapping function are $\phi^{(j)}(a, y) \leq \exp(a)$ and upper-bounds on the interval $\tau \in [0, \tau_{\max}]$ are $\phi^{(j)}(a, y) \leq \exp(a^*)$. 

\subsection{Gaussian Spike and Slab Prior}
Following \cite{Chevallier2020} and \cite{Bierkens2021}, the spike and slab prior can be simulated using PDMP dynamics. The Gaussian Spike and Slab prior has the form $\theta_j \sim w_j\mathcal{N}(0, \sigma) + (1-w_j)\delta_0$ for $j = 1,...,p$ where $w_j$ is the prior probability of inclusion and $\delta$ is the Dirac spike at zero. 

For RJPDMPs there are three types of events that occur; a reversible jump (RJ) move if a variable hits the axis ($\theta_j$ =0), events according to the likelihood and prior $\mathcal{N}(0, \sigma)$ for $\theta_j$ non-zero, or the RJ move to reintroduce a variable. The rate to reintroduce a variable is calculated in \cite{Chevallier2020} where it is found to be
$$
\frac{p_{rj}}{\sqrt{2\pi\sigma^2}}\frac{w}{(1-w)}
$$
where $p_{rj}$ is a tuning parameter for the probability of jumping to the reduced model when hitting the axis in \cite{Chevallier2020}. The time to hit an axis for parameter $\theta_j$ will be $-\theta_j/v_j$ when $\theta_jv_j > 0$ and $\infty$ otherwise. Finally simulating non RJ events is equivalent to simulating from the Gaussian slab and thus the part of the event rate contributed by the prior on parameter $\theta_j$ is $v_j(\theta_j + v_jt)/\sigma^2$ for $t\in [0,\tau_{\max}]$. This is linear in time so will be exactly simulated by the CC-PDMP approach.

\subsection{Cauchy Prior}
A Cauchy prior on $\theta$ with scale $b$ and mean $m$ has the form $p(\theta) \propto \exp(-U(\theta))$ where 
$$
U(\theta) = \log\left(1 + \frac{(\theta - m)^2}{b}\right)
$$
the term contributed to the rate depends on
$$
v\partial_{\theta}U(\theta) = v\frac{2(m-\theta)}{b + (m-\theta)^2}
$$
where $v$ is the velocity corresponding to $\theta$. Following the bounds from the Cauchy likelihood (B.1) we have 
$$
v\partial_{\theta}U(\theta + tv) \leq f(0) + f'(0)t + \frac{M}{2!}t^2
$$
where $f(0) = v\frac{2(m-\theta)}{b + (m-\theta)^2}$, $f'(0) = 2v\frac{(\theta-m)^2 - b}{(b + (\theta-m)^2)^2}$ and $M = 3$. 

\subsection{PDMP-based samplers on restricted domains}
PDMP-based samplers can be implemented on a restricted domain following \cite{bierkens2018piecewise} or when the density is only piece-wise smooth following \cite{Chevallier2021}. We consider some PDMP samplers on the restricted domain $\theta > 0$. Simulation involves tracking the time to hitting the boundary $\tau_b = \text{inf}\{t > 0 : \theta + tv = 0\}$. If an event time occurs before the boundary is hit the sampler evolves with the usual rate. If the boundary is hit before the switching event, the sampler is reflected off of the boundary by sampling a new velocity from a probability measure concentrated on directions that are normal to the boundary. We refer the reader to \cite{bierkens2018piecewise} for details. 

\subsection{Generalised inverse Gaussian prior}
The generalised inverse Gaussian prior has the form $p(\theta) \propto \exp(-U(\theta))$ where 
$$
U(\theta) = \theta + \theta^{-1} + 2\log(\theta)
$$
and $\theta >0$. The term contributed to the rate depends on
$$
v\partial_{\theta}U(\theta) = v - v\theta^{-2} + 2\frac{v}{\theta}.
$$
When $v > 0$ we have decomposition $f_u(t) = v(1 + \frac{2}{\theta + vt})$, $f_n(t) = -\frac{v}{(\theta + tv)^2}$ and $f_n'(t) = v^2\frac{2}{(\theta + vt)^3}$. \\
When $v < 0$ we have decomposition $f_u(t) = v(1 + \frac{1}{(\theta + vt)^2})$, $f_n(t) = \frac{v}{(\theta + tv)}$ and $f_n'(t) = -v^2\frac{2}{(\theta + vt)^2}$.

\subsection{Gamma prior}
The Gamma prior has the form $p(\theta) \propto \exp(-U(\theta))$ where 
$$
U(\theta) = \beta\theta - (\alpha-1)\log(\theta)
$$
and $\theta >0$ for hyper-parameters $\alpha,\beta>0$. The term contributed to the rate depends on
$$
v\partial_{\theta}U(\theta) = v\beta - v\frac{(\alpha-1)}{\theta}.
$$
When $v(\alpha -1) > 0$ we have decomposition $f_u(t) = v(\beta + \frac{(\alpha-1)}{\theta + vt})$, $f_n(t) = 0$. \\
When $v(\alpha -1) < 0$ we have decomposition $f_u(t) = v\beta$, $f_n(t) = v\frac{(\alpha-1)}{\theta + vt}$ and $f_n'(t) = -v^2\frac{(\alpha-1)}{(\theta + vt)^2}$.

\section{Comparison of CC-PDMP and CC-ARS}
We note a quick comparison between CC-PDMP and sampling directly with CC-ARS where computational cost is measured by the number of proposals used in the adaptive thinning (CC-PDMP) or adaptive sampling (CC-ARS). Both sampling methods are implemented on a generalised inverse Gaussian distribution (GIG). The unnormalised GIG density function is
$$
\pi(x) \propto \exp\left(x + x^{-1} + 2\log(x)\right),
$$
which is log concave for $2\log(x)$ and log convex for $x+x^{-1}$, see section B.6 for details on the CC-PDMP. We favour the performance of CC-ARS in the 1-dimensional setting since the method yields independent samples while CC-PDMP has correlated samples. Event times in the PDMP-based method give a continuous trajectory where samples can be taken at a consistent times along the trajectory. Rows 1 and 2 of Figure \ref{fig:arscompare} show the methods have similar performance in exploring the density. Jumping from row 2 to 3 we can see that the PDMP sampler only needs to go into the tail once to give a good approximation to the tail of the density. In higher dimensional targets CC-ARS must be used within Gibbs sampling where the efficiency of the sampler will be reduced.

\begin{figure}[H]
\centering
\includegraphics[width=.8\textwidth]{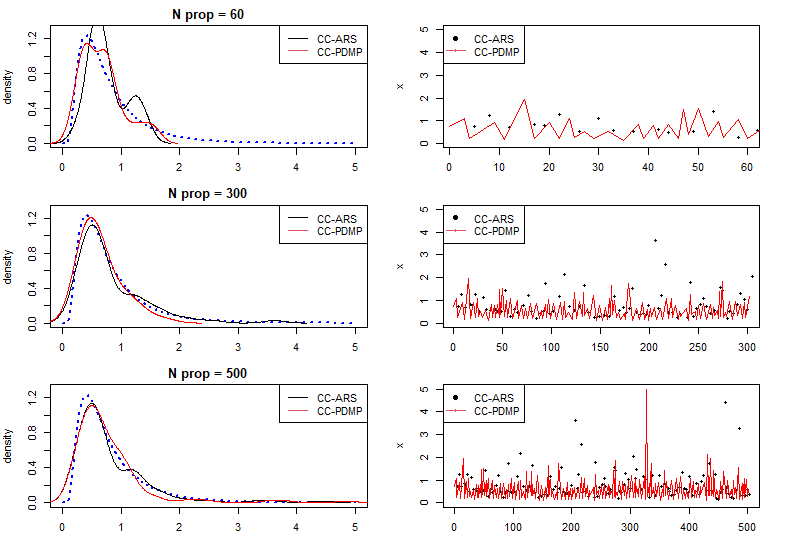}
\caption{Comparison of typical sampling dynamics for PDMP and ARS sampling using concave-convex upper-bounding in a 1-dimensional GIG distribution. Each row shows the results for an increasing number of proposals using in the sampling procedure. Continuous lines are shown for the PDMP dynamics and points for the ARS method.}
\label{fig:arscompare}
\end{figure}

%
%

\end{document}